\documentclass[10pt]{siamltex}
\usepackage{epsfig}
\usepackage{amssymb}
\usepackage{amsbsy}
\usepackage{newlfont}
\usepackage{amsmath}
\usepackage{psfrag}
\usepackage{graphicx}
\usepackage{color}  
\usepackage{xspace,float,stackrel,enumerate}
\usepackage[bbgreekl]{mathbbol}
\usepackage{remarks}

\def\Real{{\mathbb R}}

\newenvironment{Ventry}[1]%
{\begin{list}{}{%
\settowidth{\labelwidth}{{\rm (#1)}}%
\setlength{\leftmargin}{\labelwidth+\labelsep}}}
{\end{list}}

\newremark{remark}{Remark}[section]
\newcounter{itemrem}
\def\bRem{\begin{remark} \setcounter{itemrem}{0}}

\def\eRem{\end{remark}}

\def\DIV{\nabla{\cdot}}      
 
\def\GRAD{\nabla}   
\def\LAP{\Delta}

\def\ADV{{\cdot}\GRAD}       
\def\SCAL{{\cdot}}

\def\diff{\,{\textrm d}}

\def\R{\Real}

\newcommand{\ie}{i.e.,\@\xspace}
\newcommand{\eg}{e.g.,\@\xspace}

\newcommand{\bVentry}[1]{\begin{Ventry}{#1}}
\newcommand{\eVentry}{\end{Ventry}}

\catcode`\@=11
\newdimen\linespacing
\linespacing=\baselineskip
\newcommand\subsubsectionmodif{\@startsection{subsubsection}{3}%
  {\z@}%
  {.5\linespacing\@plus.7\linespacing}
  {-.5em}
  {\normalfont\normalsize\subsecstyle
  \rightskip=\z@ \@plus 8em\pretolerance=10000}}

\floatstyle{ruled}
\newfloat{algo}{ht}{algo}
\floatname{algo}{Algorithm}
\newcommand{\balgo}{\begin{algo}}
\newcommand{\ealgo}{\end{algo}}

\newlength\pageone
\newlength\pagetwo
\newlength\interspace

\newlength\retraitspace
\newlength\calculspace

\def\bi{\begin{itemize}}
\def\ei{\end{itemize}}
\def\berom{\begin{enumerate}[{\rm(i)}]}
\def\eerom{\end{enumerate}}
\def\bproof{\begin{proof}}
\def\eproof{\end{proof}}

\newcommand{\Question}[1]{\marginpar{}}
\renewcommand{\Question}[1]{%
            \marginpar{\flushleft\scriptsize\bfseries\upshape#1}}
\catcode`\@=11
\@addtoreset{equation}{section}

\newcommand{\bfjlg}{\boldsymbol}         

\newcommand{\ba}{{\bfjlg a}}
\newcommand{\bb}{{\bfjlg b}}

\newcommand{\be}{{\bfjlg e}}
\newcommand{\bef}{{\bfjlg f}}

\newcommand{\bh}{{\bfjlg h}}

\newcommand{\bl}{{\bfjlg l}}
\newcommand{\bm}{{\bfjlg m}}

\newcommand{\bq}{{\bfjlg q}}

\newcommand{\bu}{{\bfjlg u}}

\newcommand{\bx}{{\bfjlg x}}

\newcommand{\bE}{{\bfjlg E}}
\newcommand{\bF}{{\bfjlg F}}

\newcommand{\bJ}{{\bfjlg J}}

\newcommand{\bU}{{\bfjlg U}}

\newcommand{\bX}{{\bfjlg X}}
\newcommand{\bY}{{\bfjlg Y}}

\newcommand{\calC}{{\mathcal C}}

\newcommand{\calE}{{\mathcal E}}

\newcommand{\polg}{{\mathbb g}}
\newcommand{\polh}{{\mathbb h}}

\newcommand{\polG}{{\mathbb G}}
\newcommand{\polH}{{\mathbb H}}
\newcommand{\polI}{{\mathbb I}}

\newcommand{\polM}{{\mathbb M}}
\newcommand{\polN}{{\mathbb N}}

\newcommand{\polQ}{{\mathbb Q}}
\newcommand{\polR}{{\mathbb R}}
\newcommand{\polS}{{\mathbb S}}

\usepackage[numbers]{natbib}
\let\cite=\citet

\begin{document}

\newcommand\footnotemarkfromtitle[1]{%
\renewcommand{\thefootnote}{\fnsymbol{footnote}}%
\footnotemark[#1]%
\renewcommand{\thefootnote}{\arabic{footnote}}}

\title{Viscous
  regularization of the Euler equations and entropy
  principles\footnotemark[1]}

\author{Jean-Luc Guermond\footnotemark[2] \and Bojan Popov\footnotemark[3]}

\date{Draft version \today}

\maketitle

\renewcommand{\thefootnote}{\fnsymbol{footnote}} \footnotetext[1]{
  This material is based upon work supported in part by the National
  Science Foundation grants DMS-1015984 and DMS-1217262, by the Air
  Force Office of Scientific Research, USAF, under grant/contract
  number FA9550-09-1-0424, FA99550-12-0358, and by Award No.~KUS-C1-016-04, made by
  King Abdullah University of Science and Technology (KAUST).  Draft
  version, \today} \footnotetext[2]{Department of Mathematics, Texas
  A\&M University 3368 TAMU, College Station, TX 77843, USA. On leave from CNRS, France.}
\footnotetext[3]{Department of Mathematics, Texas
  A\&M University 3368 TAMU, College Station, TX 77843, USA.}
\renewcommand{\thefootnote}{\arabic{footnote}}

\begin{abstract} 
This paper investigates a general class of viscous regularizations of
the compressible Euler equations. A unique regularization is
identified that is compatible with all the generalized entropies \`a
la \cite{MR1655839} and satisfies the minimum entropy principle.
A connection with a recently proposed phenomenological model
by \cite{Brenner2006190} is made.
\end{abstract}



\pagestyle{myheadings} \thispagestyle{plain} 
\markboth{J.L. GUERMOND, B. POPOV}{Viscous regularization and entropy principles}

\section{Introduction}
Proving positivity of the density and internal energy and proving a
minimum principle on the specific entropy of numerical approximations
of the compressible Euler equations is a challenging task that has so
far been achieved for very few numerical schemes on arbitrary meshes
in two and higher space dimensions.  The Godunov scheme
(\cite{MR0119433}) and some variants of the Lax\footnote{The Lax
  scheme is often called the Lax-Friedrichs scheme in the literature.}
scheme (\cite{MR0066040}) are known to satisfy all these properties,
(see \cite{MR1094256} for the Godunov scheme,
\cite[Appendix]{MR1379283} for the explicit Lax algorithm, and
\cite{MR1829553} for the implicit version of the Lax algorithm). The
argumentation for the Godunov scheme relies on the fact that Riemann
problems are solved exactly at each time step and averaging Riemann
solutions preserves the above mentioned properties. None of the above
arguments can be readily extended to central high-order schemes and
more generally to schemes that are based on Galerkin
approximations. One way to address this issue consists of using the
standard parabolic regularization of the Euler equations to construct
a scheme for which the vanishing viscosity is proportional to the mesh
size.  The problem with this approach is that the regularization acts
on the conserved variables which are the density, momentum, and total
energy. Since the momentum and total energy are not Galilean
invariant, a change of reference frame by translation and/or rotation
changes the regularization.  A way out of this dilemma consists of
considering the Navier-Stokes regularization as a starting point to
construct a numerical method. However, one then encounters two serious
difficulties. The first one is that the Navier-Stokes equations do not
include any regularization in the continuity equation, which is
inconsistent with most numerical discretizations.  The second one is
that whereas it is known that the Euler equations satisfy a minimum
entropy principle on the specific entropy (see \eg \cite{MR863987}),
it is also known that the Navier-Stokes equations violate this minimum
principle if the thermal diffusivity is nonzero, see \eg \cite[Thm
  8.2.3]{MR1459989}. These two observations make the Navier-Stokes
regularization inconvenient for numerical purposes. One is then lead
to ponder on the following question: Is it possible to find a
regularization of the Euler equations that is Galilean invariant,
ensures positivity of the density and internal energy, satisfies a
minimum entropy principle, and is compatible with a large class of
entropies inequalities? The objective of this paper is to answer to
this question.

The paper is organized as follows.  The parabolic and the
Navier-Stokes regularizations and their apparent shortcomings
mentioned above are discussed in \S\ref{Sec:LX_NS}. A general family
of regularizations is introduced and investigated in
\S\ref{Sec:new_regularization} and \S\ref{Sec:Entropy}. The minimum
entropy principle is investigated in \S\ref{Sec:new_regularization}
and the compatibility with entropy inequalities is studied in
\S\ref{Sec:Entropy}.  The key result of this paper is
Theorem~\ref{Thm:entropy_inequality}: only one regularization
technique satisfies the minimum entropy principle and is compatible
with all the generalized entropies of \cite{MR1655839}.  This
formulation is compared in \S\ref{Sec:Brenner} with a reformulation of
the Navier-Stokes equations proposed by \cite{Brenner2006190} that is
based on heuristic arguments. A striking observation is that by
distinguishing the so-called mass and volume velocities, it is
possible to re-write the proposed regularization into a form similar
to that of the Navier-Stokes equations with rotation invariant viscous
fluxes.  This way of looking at the regularization reconciles the
parabolic and Navier-Stokes regularizations and shows that they are
two faces of the same coin. The key results of the paper are
summarized in \S\ref{Sec:Conclusions}. Standard identities and
inequalities from thermodynamics that are used in this paper are
collected in Appendix~\ref{Sec:Appendix}.

\section{Standard regularizations} \label{Sec:LX_NS} We review in this
section some well-known regularization techniques and discuss the pros
and cons thereof.
\subsection{Statement of the problem}
Consider the compressible Euler equations in conservative form in
$\R^d$,
\begin{align}\label{eq:Euler}
  &\partial_t \rho + \DIV \bm = 0, \\
  &\partial_t \bm + \DIV (\bu\otimes \bm) + \GRAD p  = 0, \label{eq:Euler_momentum}\\
  &\partial_t E + \DIV (\bu(E+p)) = 0,\label{eq:Euler_energy}\\
  & \rho(\bx,0)=\rho_0,\qquad \bm(\bx,0)=\bm_0,\qquad E(\bx,0)=E_0, 
\label{eq:Euler_init}
\end{align}
where the dependent variables are the density, $\rho$, the momentum,
$\bm$ and the total energy, $E$. We adopt the usual convention that
for any vectors $\ba$, $\bb$, with entries $\{a_i\}_{i=1,\dots,d}$,
$\{b_i\}_{i=1,\ldots,d}$, the following holds: $(\ba\otimes\bb)_{ij} =
a_ib_j$ and $\DIV\ba = \partial_{x_j} a_{j}$,
$(\GRAD\ba)_{ij}=\partial_{x_i} a_j$. Moreover, for any order 2 tensors
$\polg$, $\polh$, with entries $\{g_{ij}\}_{i,j=1,\dots,d}$,
$\{h_{ij}\}_{i,j=1,\dots,d}$, we define $(\DIV\polg)_j =
\partial_{x_i} g_{ij}$, $\ba\SCAL\GRAD = a_i\partial_{x_i}$,
$(\polg\SCAL \ba)_i= g_{ij} a_j$, $\polg{:}\polh = g_{ij} \polh_{ij}$
where repeated indices are summed from $1$ to $d$.

The pressure, $p$, is given by the equation of state which we assume
to derive from a specific entropy, $s(\rho,e)$, through the
thermodynamics identity:
\begin{equation} 
T \diff s := \diff e + p \diff \tau,
\end{equation} 
where $\tau:=\rho^{-1}$, $e := \rho^{-1}E - \frac12 \bu^2$ is the
specific internal energy, $\bu:=\rho^{-1} \bm$ is the velocity of the
fluid particles.  For instance it is common to take
$s=\log(e^{\frac{1}{\gamma-1}}\rho^{-1})$ for an polytropic ideal gas.  Using the
notation $s_e:=\frac{\partial s}{\partial e}$ and $s_\rho
:=\frac{\partial s}{\partial \rho}$, this definition implies that
\begin{align}
&s_e := T^{-1}, \qquad
s_\rho := - p T^{-1} \rho^{-2}. \label{s_rho_s_e}
\end{align}
The equation of state takes the form
$p:= - \rho^2 s_\rho s_e^{-1}$, or
\begin{align}
&ps_e + \rho^2 s_\rho =0. \label{pressure_elimination}
\end{align}
The key structural assumption is that $-s$ is strictly convex with
respect to $\tau:=\rho^{-1}$ and $e$.  Upon introducing
$\sigma(\tau,e) := s(\rho,e)$, the convexity hypothesis is equivalent
to assuming that $\sigma_{\tau \tau} \le 0$, $\sigma_{ee} \le 0$, and
$\sigma_{\tau \tau} \sigma_{ee} - \sigma_{\tau e}^2 \le 0$ (see \eg 
\cite{MR1410987}).  This in
turn implies that
\begin{align}
  \partial_\rho (\rho^2 s_{\rho})< 0 ,\qquad s_{ee} < 0,
  \qquad 0 < \partial_\rho (\rho^2 s_{\rho}) s_{ee} - \rho^2
  s_{\rho e}^2, \label{convexity}
\end{align}
or equivalently that the following matrix
\begin{equation}
\Sigma := \left(\begin{matrix}\rho^{-1} \partial_\rho (\rho^2 s_{\rho}) & \rho s_{\rho e}\\ 
\rho s_{\rho e} & \rho s_{ee}\end{matrix}\right), \label{def_of_Sigma}
\end{equation}
is negative definite. In the rest of the paper we assume 
that \eqref{convexity} holds and the temperature be positive
\begin{equation}
  0 < s_e. \label{s_e_positivity}
\end{equation}

\begin{remark} 
  Note in passing that contrary to what is sometimes done in the
  literature, we do not assume that the pressure be positive, which
  requires $s_\rho<0$ (see \eg \cite[p.~99]{MR1410987},
  \cite[(2.3)]{MR1655839}). For instance, the assumptions
  \eqref{convexity} and \eqref{s_e_positivity} hold for stiffened
  gases, but the quantity $s_\rho$ can change sign. It is shown in the
  Appendix (see Remark~\ref{Rem:hyperbolicity}) that the convexity
  assumption \eqref{convexity} and the positivity of the
  temperature~\eqref{s_e_positivity} are sufficient to prove that the
  Euler system is hyperbolic. This fact was first established by
  \cite{MR0122351} in one dimension. It was established again
in \cite{MR0285799} and \cite{MR1655839}.
\end{remark}

The objective of the present paper is to introduce a viscous
regularization of \eqref{eq:Euler}--\eqref{eq:Euler_init} that is
compatible with thermodynamics and that can serve as a reasonable
starting point for numerical approximation.

\subsection{Monolithic parabolic regularization} \label{Sec:Para}
A common regularization of \eqref{eq:Euler} for theoretical and
numerical purposes consists of the following monolithic parabolic
regularization:
\begin{align}\label{eq:Para}
  &\partial_t \rho + \DIV \bm = \epsilon \LAP \rho, \\ \label{eq:Para_moment}
  &\partial_t \bm + \DIV (\bu\otimes \bm) + \GRAD p  = \epsilon \LAP m , \\ \label{eq:Para_energy}
  &\partial_t E + \DIV (\bu(E+p)) = \epsilon \LAP E,\\
  &\rho(\bx,0)=\rho_0,\qquad \bm(\bx,0)=\bm_0,\qquad E(\bx,0)=E_0, 
\label{eq:Para_init}
\end{align}
where $\epsilon$ is a small parameter.  We call this regularization
monolithic since no distinction is made between the conserved
quantities, \ie the operator $\epsilon \LAP$ is applied blindly to all
the conserved quantities.  

It can be shown that the Lax-Friedrichs scheme and its parabolic analog introduced 
in \cite{MR1379283}  are approximations of \eqref{eq:Para}.  For instance, considering a nonlinear
conservation equation $\partial_t\bU +\DIV\bF(\bU)=0$, where $\bU$ is
the dependent vector-valued variable in $\Real^m$, the scheme
introduced in \cite[p.163]{MR0066040} in one space dimension consists
of considering
\begin{align}
\bU_i^{n+1} & = \frac12(\bU_{i+1}^n+\bU_{i-1}^n) - \frac12 \lambda
(\bF(\bU_{i+1}^n) - \bF(\bU_{i-1}^n)) \label{LX} \\ 
&=\bU_i^n - \frac12 \lambda
(\bF(\bU_{i+1}^n) - \bF(\bU_{i-1}^n)) + \tau \frac12 h^2\tau^{-1}
\frac{(\bU_{i+1}^n-2\bU_i^n+\bU_{i-1}^n)}{h^2}, \nonumber
\end{align}
where $h$ is the mesh size, $\tau$ is the time step, and
$\lambda:=\tau h^{-1}$. Assuming the flux $\bF$ to be uniformly
Lipschitz, to simplify, and upon introducing the maximum wave speed
$\beta := \|\bF'\|_{L^\infty{(\Real^m;\Real^m{\times}\Real^m)}}$ and
the CFL number $\text{cfl}:=\beta \tau h^{-1}$, \eqref{LX} is the
centered second-order approximation of the following parabolic
regularization of the conservation equation $ \partial_t\bU
+\DIV\bF(\bU) - \epsilon \LAP \bU =0, $ with the artificial viscosity
$\epsilon :=\frac12 h \lambda^{-1} = \frac{1}{\text{cfl}}\frac12 \beta
h$. In other words, the Lax-Friedrichs scheme is a centered
second-order approximation of \eqref{eq:Para}-\eqref{eq:Para_init}
with the numerical viscosity $\epsilon = \frac{1}{\text{cfl}} \frac12
h \||\bu|+c\|_{L^\infty(\polR^d{\times}\polR_+)}$, where $c$ is the
speed of sound.  That the CFL number appears at the denominator of the
artificial viscosity makes this scheme over-dissipative. It is often
more appropriate to consider the following alternative
\begin{align*}
\bU_i^{n+1} & = \bU_i^n  - \frac12 \lambda (\bF(\bU_{i+1}^n) - \bF(\bU_{i-1}^n)) +  
\frac12 \lambda |\beta| h^2 \frac{(\bU_{i+1}^n-2\bU_i^n+\bU_{i-1}^n)}{h^2},
\end{align*}
which is also a centered second-order approximation of the parabolic
regularization $\partial_t\bU +\DIV\bF(\bU) - \epsilon \LAP \bU =0$
with the viscosity $\frac12 \beta h$, which is more traditionally
associated with up-winding. This algorithm is often abusively referred
to as the Lax-Friedrichs scheme. Both the above numerical schemes have
interesting positivity and entropy properties, see \eg
\cite{MR0393870,MR863987,MR2249160}, \cite{MR1379283}.

Despite its appealing mathematical properties, the above
regularization is often criticized by physicists since it seemingly
violates the Galilean and rotation invariance. It also dissipates the
density, the momentum and the total energy, which seemingly are again
aberrations from the physical point of view. When looking at
\eqref{eq:Para}-\eqref{eq:Para_init}, it is indeed difficult to see
how this set of equations can be reconciled with the Navier-Stokes
equations which are usually viewed by physicists to be the acceptable
regularization of the Euler equations.

\subsection{Navier-Stokes regularization} \label{Sec:NS}
As mentioned above, a common ``physical'' way to regularize the
Euler system \eqref{eq:Euler}-\eqref{eq:Euler_init} consists of
considering this system as the limit of the Navier-Stokes equations
\begin{align}\label{eq:NS}
  &\partial_t \rho + \DIV\bm = 0, \\
  &\partial_t \bm + \DIV (\bu\otimes \bm) + \GRAD p - \DIV\polg = 0, \\
  &\partial_t E + \DIV (\bu(E+p)) - \DIV (\bh + \polg\SCAL \bu) = 0, \\
  &\rho(\bx,0)=\rho_0,\qquad \bm(\bx,0)=\bm_0,\qquad E(\bx,0)=E_0. 
\label{eq:NS_init}
\end{align}
where $\polg$ and $\bh$ are the viscous and thermal fluxes. The most
elementary model compatible with Galilean invariance consists of
assuming that
\begin{equation}
\polg = 2\mu \GRAD^s\bu + \lambda \DIV\bu \polI,
\qquad \bh=\kappa \GRAD T. \label{def_NS_regul}
\end{equation}
where $\GRAD^s\bu:=\mu (\GRAD \bu + (\GRAD\bu)^T)$, $\polI$ is the
identity matrix in $\Real^d$, and $T$ is the temperature, $T:=
s_e^{-1}$.  The viscosity $\mu$ and the thermal diffusivity $\kappa$
are required to be non-negative by the Clausius-Duhem inequality,
although these two parameters may depend on the state $(\rho,e)$.  

We claim that \eqref{eq:NS}-\eqref{def_NS_regul} is not appropriate
for numerical purposes and we identify at least two obstructions.  The
first problem is that the minimum entropy principle cannot be
satisfied for general initial data if the thermal dissipation is not
zero. More precisely, assuming $\kappa\neq 0$, for any $r\in \Real$,
there exist initial data so that the set $\{s\ge r\}$ is not
positively invariant.  Let us recall a simple proof of this statement
borrowed from \cite[Thm 8.2.3]{MR1459989}. The specific entropy for
the Navier-Stokes system satisfies
\begin{equation}
  \partial_t s + \bu\ADV s = \frac{1}{\rho T}\left(\polg{:}\GRAD^s\bu + 
    \DIV(\kappa\GRAD T)\right).
\end{equation}
Assume that $\bu_0:=\bm_0 \rho_0^{-1}$ is constant. Assume also that
the equation of state of the fluid is such that $p_e\neq 0$, then one
can use $T$ and $s$ as independent state variables since $\rho^2
\text{det}\left(\frac{D(T,s)}{D(\rho,e)}\right) =
\frac{\rho^2}{s_e^2}(s_\rho s_{ee} - s_e s_{\rho e}) = p_e \neq 0$
(see \eqref{thermodynamic_compatibility}).  One can then choose $s_0$
with global minimum at $0$ and $T_0$ so that $\LAP T_0(0)<0$ and
$\GRAD T_0(0)=0$. Without loss of generality, we assume that
$\kappa>0$ in a neighborhood of $0$. Then $\partial_t s(0,0) = \kappa
\rho_0^{-1}(0) \LAP T_0(0) < 0$, thereby proving that $\{s\ge r\}$ is
not positively invariant for the regularized system
\eqref{eq:NS}--\eqref{def_NS_regul}.

Another argument often invoked against the presence of thermal
dissipation is that it is incompatible with symmetrization of the
Navier-Stokes system when using the generalized entropies of Harten
for polytropic ideal gases.  (The function $\rho f(s)$ is said to be a
generalized entropy if $f'\gamma^{-1} - f'' > 0$, $f'>0$ and
$f\in\calC^2(\Real;\Real)$, see \cite{MR694161}.)  It is proved in
\cite{MR831553} that the only generalized entropy that symmetrizes the
Navier-Stokes system \eqref{eq:NS}--\eqref{def_NS_regul} is the
trivial one $\rho s$ when $\kappa\not=0$, see also \cite[ (2.11) and
Remark 2, page 460]{MR2249160}. Note though that symmetrization of the
viscous fluxes is not necessary to prove entropy dissipation. It is
nevertheless true that the Navier-Stokes system with $\kappa\not=0$
does not admit a generalized entropy inequality if $f''(s)\not=0$, and
this fact is a consequence of the following quadratic form not being
non-negative: $f'(s) X^2 - f''(s) XY$, $(X,Y)\in \Real^2$. Symmetry of
the viscous flux is not a necessary condition for entropy dissipation,
see \eg \cite[\S1.1]{MR2857012}.

The above two arguments seem to imply that one should take $\kappa=0$
if one wants to use the Navier-Stokes system as a numerical device
that regularizes the Euler equations, satisfies the minimum entropy
principle, and satisfies entropy inequalities. In that case, one then
faces a serious obstruction when solving for contact waves. For
instance assuming that the initial data, $\rho_0$, $\bm_0$, $\bE_0$
are such that the exact velocity is constant in time and space, say
$\bu=\beta \be_x$, the problem \eqref{eq:NS}--\eqref{eq:NS_init}
reduces to solving two linear transport equations
\begin{align}
\partial_t \rho + \beta \partial_x \rho=0,\quad \rho(\cdot,0)=\rho_0,\\
\partial_t E + \beta \partial_x E =0, \quad
E(\cdot,0)=E_0.
\end{align}
Note that $\bu$ being constant implies that the pressure gradient is
zero. The exact solution is $\rho(\bx,t)=\rho_0(\bx- \beta t \be_x)$.
To make things a little bit more interesting assume that $\rho_0$ is
piecewise constant, say $\rho_0(x)=1$ if $x<0$ and $\rho_0(x)=2$ if
$x>0$.  In the absence of some sort of regularization, the above two
linear transport equations are difficult to solve numerically.  Except
for the method of characteristics and Lagrangian based techniques, we
are not aware of any numerical methods that can solve these equations
without resorting to some kind of viscous regularization.

In conclusion, if positivity of the density, the minimum entropy
principle and a reasonable approximation of contact
discontinuities is desired, the Navier-Stokes regularization does not
seem to be appropriate to regularize
\eqref{eq:Euler}--\eqref{eq:Euler_init}, whether $\kappa$ is zero or
not.

\section{General regularization} \label{Sec:new_regularization}
We investigate in this section the properties of a class of
regularizations that we expect to be as general as possible. More
precisely, let us consider the following general regularization for
the Euler system:
\begin{align}
  &\partial_t \rho + \DIV \bm - \DIV \bef = 0, \label{eq:euler_regul}\\
  &\partial_t \bm + \DIV (\bu\otimes \bm) + \GRAD p- \DIV\polg = 0, \label{eq:euler_regul_moment}\\
  &\partial_t E + \DIV (\bu(E+p)) - \DIV (\bh + \polg\SCAL \bu) =
  0,\label{eq:euler_regul_energy}
\end{align}
where for the time being we let the fluxes $\bef$, $\polg$, and $\bh$
to be as general as possible.  A theory of viscous regularization for
general nonlinear hyperbolic system has been developed in
\cite{MR2857012} and \cite[Chap~6]{MR1459989}. This theory identifies
classes of entropy-dissipative viscous regularizations and establishes
short term existence results.  Our objective in this paper is more
restrictive. We want to construct the fluxes $\bef$, $\polg$, and
$\bh$ so that \eqref{eq:euler_regul}-\eqref{eq:euler_regul_energy}
gives a positive density, gives a minimum principle on the specific
entropy, and is compatible with a large class of entropies. (Note in
passing that the positivity of the internal energy will be a
consequence of the positivity of the density and the minimum entropy
principle.) In the rest of the paper, we are going to work under the
assumption that \eqref{eq:euler_regul}-\eqref{eq:euler_regul_energy}
has a smooth solution.

\subsection{Positivity of the density}
Let us now choose the flux $\bef$ so that it regularizes the mass
conservation equation. From the theory of second-order elliptic
equation we conjecture that $a(\rho,e) \GRAD\rho$ should be
appropriate, where $a(\rho,e)$ is a smooth positive function of $\rho$
and $e$. In particular, it is reasonable to expect that the following
choice implies positivity of the density:
\begin{equation}
a(\rho,e) = \chi(\rho,e)\varphi'(\rho),
\end{equation}
where $\chi$ is a smooth positive function of $\rho$ and $e$ and
$\varphi$ is a strictly increasing function. This definition gives
$\bef = \chi(\rho,e)\GRAD\varphi(\rho)$.  This regularization is at
least compatible with the positive density principle as stated in the
following.
\begin{lemma}[Positive Density
  Principle] \label{Lem:positive_density_principle} Let $\bef
  =a(\rho,e) \GRAD\rho$ in \eqref{eq:euler_regul}, with $a\in
  L^\infty(\Real^2;\Real)$ and $\inf_{(\xi,\eta)\in \Real^2} a(\xi,\eta)
  >0$.  Assume that $\bu$ and $\DIV\bu \in
  L^\infty(\Real^d{\times}\Real_+;\Real)$.  Assume also that there are
  constant states at infinity $\rho^\infty$, $\bu^\infty$, so that the
  supports of $\rho(\cdot,\cdot)-\rho^\infty$ and
  $\bu(\cdot,\cdot)-\bu^\infty$ are compact in $\Real^d{\times}(0,t)$,
  for any $t>0$. Assume finally that $\rho_0 - \rho_\infty \in
  L^2(\Real^d;\Real)$.  Then the solution of \eqref{eq:euler_regul} is
  such that
\begin{equation}
\stackrel[\bx \in \Real^d]{}{\text{\em ess inf}} \rho(\bx,t) \ge 0,
\qquad \forall t\ge 0.
\end{equation}
\end{lemma}%
\begin{proof}
  Owing to the assumed regularity of $\bu$ and $\rho_0$, the theory of
  parabolic equations implies that there is a unique solution to
  \eqref{eq:euler_regul} such that $\rho - \rho_\infty \in
  L^\infty((0,\infty);L^2(\Real^d)) \cap L^2((0,\infty);H^1(\Real^d))$
  and $\partial_t \rho \in L^2((0,\infty);H^{-1}(\Real^d))$, see \eg
  \cite[p.356]{MR1625845}.

  Let $\epsilon>0$ and let $h_\epsilon(x)$ be a smooth concave
  function that approximates $\min(x,0)$ uniformly over $\Real$; say
  there is $c>0$ so that $\sup_{s\in
    \Real}|h_\epsilon(s)-\min(s,0)|+|h_\epsilon(s)-s h_\epsilon'(s)|<
  c \epsilon$ and $h''\le 0$.  Let $t>0$ be some fixed time. Let
  $B(0,R)$ be the ball centered at $0$ of radius $R$ such that the
  supports of $\rho(\cdot,\tau)-\rho^\infty$ and
  $\bu(\cdot,\tau)-\bu^\infty$ are in $B(0,R)$ for all $\tau\in
  [0,t]$. Let $\chi$ be a regularized characteristic function with the
  following properties: $\chi|_{B(0,R)}=1$ and
  $\chi|_{\Real^d\setminus B(0,R+1)}=0$.  Multiplying the weak form of
  \eqref{eq:euler_regul} by the legitimate test function $\chi
  h_\epsilon'(\rho)$ we obtain
\[
\int_{\Real^d}\Big( (\partial_t h_\epsilon(\rho) + \bu\GRAD
h_\epsilon(\rho) + \rho h_\epsilon'(\rho) \DIV\bu)\chi(\bx) + a\GRAD
\rho \GRAD (\chi h_\epsilon'(\rho)) \Big)\diff\bx =0.
\]
Using that the properties of $\chi$, we simplify the above equation as
follows:
\begin{align*}
\int_{\Real^d}\Big( \partial_t h_\epsilon(\rho) + \bu\GRAD
h_\epsilon(\rho) + \rho h_\epsilon'(\rho) \DIV\bu  +
h_\epsilon''(\rho)a |\GRAD\rho|^2\Big)\diff\bx & =0\\
\int_{\Real^d}\Big( \partial_t h_\epsilon(\rho) +
\DIV(h_\epsilon(\rho)\bu)  + (\rho h_\epsilon'(\rho)-h_\epsilon(\rho)) \DIV\bu  +
h_\epsilon''(\rho)a |\GRAD\rho|^2\Big)\diff\bx & =0
\end{align*}
Now, we integrate over time 
and, owing to the assumptions regarding the behavior of $\bu$,
$\rho$ and $a$, we obtain
\begin{align*}
  \int_{\Real^d} h_\epsilon(\rho(\bx,t))\diff\bx  
& \ge - \int_0^t \int_{\Real^d}|\rho
  h_\epsilon'(\rho) -h_\epsilon(\rho)| |\DIV\bu| \diff\bx \diff t+ \int_{\Real^d}
  h_\epsilon(\rho_0(\bx)) \\
& \ge -c \epsilon + \int_{\Real^d}
  h_\epsilon(\rho_0(\bx)).
\end{align*}
We can now pass to the limit on $\epsilon$ using the Lebesgue
dominated convergence and we obtain $\int_{\Real^d}
\min(\rho(\bx,t),0) \ge 0$.  The result follows readily.
\end{proof}

\subsection{Minimum entropy principle} We now investigate under which
conditions on the fluxes $\bef$, $\polg$ and $\bh$, a minimum
principle on the specific entropy holds.  In order to account for
impact of the viscous part in the mass conservation, we change the
notation of the various viscous fluxes as stated in the following
lemma.
\begin{lemma} \label{lem:entropy_equation} 
Setting $\polg = \polG + \bef\otimes\bu,$ and $\bh = \bl -
\frac12\bu^2 \bef$, the specific entropy for the system
\eqref{eq:euler_regul}--\eqref{eq:euler_regul_energy} satisfies
\begin{equation} \label{entropy_step_one}
  \rho (\partial_t s + \bu \SCAL \GRAD s)  + \DIV((e s_e -\rho
  s_\rho)\bef - s_e \bl)
  -\bef\SCAL\GRAD(e s_e -\rho s_\rho) 
     + \bl\SCAL\GRAD s_e - s_e \polG{:}\GRAD\bu=0.
\end{equation}
\end{lemma}%
\begin{proof}
We re-write \eqref{eq:euler_regul}--\eqref{eq:euler_regul_energy} in
non-conservative form as follows:
\begin{align*}
  &\partial_t \rho + \bu\ADV \rho + \rho\DIV \bu  - \DIV \bef =0, \\
  &\rho(\partial_t \bu + \bu\ADV \bu) + \bu \DIV \bef + \GRAD p- \DIV\polg = 0, \\
  &\rho(\partial_t\calE + \bu\ADV \calE) + \calE \DIV \bef + \DIV (\bu
  p) - \DIV (\bh + \polg\SCAL \bu) = 0.
\end{align*}
where we have defined $\calE=\rho^{-1}E$.  Then we obtain the equation
controlling the internal energy, $e=\calE-\frac12 \bu^2$, by
multiplying the momentum equation by $\bu$ and subtracting the result
from the total energy equation:
\begin{equation*}
  \rho (\partial_t e + \bu \SCAL \GRAD e) + (e - \tfrac12 \bu^2) \DIV\bef +
  p\DIV\bu - \DIV \bh - \polg{:}\GRAD \bu  = 0,
\end{equation*}

The key to obtain the equation that controls the entropy is to
multiply the mass conservation by $\rho s_\rho$, multiply the internal
energy balance by $s_e$, and add the two resulting equations.  This
linear combination is motivated by the following observation
$\partial_\alpha s = s_\rho \partial_\alpha\rho + s_e \partial_\alpha e$
which holds for any independent variable $\alpha\in \{t,\bx\}$.  We
then obtain
\begin{align*}
  \rho (\partial_t s + \bu \SCAL \GRAD s) + s_e (e - \tfrac12 \bu^2)
  \DIV\bef
  & + (ps_e + \rho^2 s_\rho) \DIV\bu \\
  & - s_e (\DIV \bh + \polg\SCAL \GRAD \bu) -\rho s_\rho \DIV \bef = 0
\end{align*}
The definition of the pressure implies that the quantity $ps_e +
\rho^2 s_\rho$ is zero, see \eqref{pressure_elimination}. This
simplification yields
\begin{align*}
  \rho (\partial_t s + \bu \SCAL \GRAD s) + (e s_e -\rho
  s_\rho)\DIV\bef &- s_e (\polg{:} \GRAD \bu) - s_e \tfrac12 \bu^2
  \DIV\bef
  - s_e \DIV \bh = 0. 
\end{align*}
We now regroup the terms 
\begin{align*}
  \rho (\partial_t s + \bu \SCAL \GRAD s) + (e s_e -\rho s_\rho
  )\DIV\bef - s_e
  \DIV(\bh + \tfrac12\bu^2 \bef)
   - s_e (\polg{:}\GRAD \bu - (\bef\otimes\bu){:}\GRAD\bu) =0,
\end{align*}
and conclude by using the definitions $\polg = \polG(\GRAD^s\bu) +
\bef\otimes\bu$ and $\bh = \bl - \frac12\bu^2 \bef$.
\end{proof}

From now on we assume that the following structure holds
for the viscous fluxes introduced in \eqref{eq:euler_regul}--\eqref{eq:euler_regul_energy}:
\begin{align}
  \polg  = \polG(\GRAD^s \bu) + \bef\otimes \bu, \qquad 
  \bh = \bl -\tfrac12 \bu^2 \bef,\qquad
 \polG(\GRAD^s \bu){:}\GRAD\bu \ge 0. \label{def_ghl}
\end{align}
We also assume that $\bef$ has  the following form:
\begin{align}
  \bef &= a(\rho,e) \GRAD \rho && a(\rho,e) \ge 0, \label{def_f}
\end{align}
and  $\bl$ is defined so that
\begin{align}
\bl = s_e^{-1}(e s_e -\rho s_\rho)\bef
+ d(\rho,e)\rho s_e^{-1} \GRAD s, && d(\rho,e) \ge 0.  \label{def_l}
\end{align}

\begin{remark}
  The conditions $\polG(\GRAD^s \bu){:}\GRAD\bu \ge 0$, $a(\rho,e) \ge
  0$, and $d(\rho,e) \ge 0$ are essential to establish the minimum
  principle on the specific entropy and the entropy inequalities (see
  Theorem~\ref{Thm:minimum_entropy_principle} and
  Theorem~\ref{Thm:entropy_inequality}).
\end{remark}

\begin{remark}
  The structural assumption $\bl = s_e^{-1}(e s_e -\rho s_\rho)\bef +
  d(\rho,e)\rho s_e^{-1} \GRAD s$ is crucial.  This condition is
  equivalent to assuming that the conservative term in
  \eqref{entropy_step_one} is of the following form: $(e s_e -\rho
  s_\rho)\bef -s_e \bl = -\DIV(d\rho\GRAD s)$.  The
  definition of $\bl$ makes sense since thermodynamics requires that
  $s_e =T^{-1} > 0$, (see \eqref{s_e_positivity}).  Note that given
  \eqref{def_f} the following alternative forms hold $\bl=(d-a)\rho
  s_\rho s_e^{-1} \GRAD\rho + a e\GRAD\rho +d\rho \GRAD e$, or
  $\bl=(a-d)(p\rho^{-1}+e)\GRAD\rho + d\GRAD(\rho e)$.
\end{remark}

Let us define the quantity
\begin{equation}
  J:=-\bef\SCAL\GRAD(e s_e -\rho s_\rho) + \bl\SCAL\GRAD s_e 
+ a \GRAD \rho \SCAL\GRAD s \label{def_of_J}
\end{equation}
which is a quadratic form with respect to $\GRAD \rho$ and
$\GRAD e$ and whose coefficients depend on $\rho$, $e$, $a(\rho,e)$,
$c(\rho,e)$, and $d(\rho,e)$.

Let $\polI_d$ be the $d{\times}d$ identity matrix.  For any symmetric
$2{\times}2$ block matrix $\polN$
\begin{align*}
\polN= & \left(\begin{matrix} n_{11} \polI_d& n_{12}\polI_d \\ 
n_{12}\polI_d & n_{22}\polI_d\end{matrix}\right) \quad \text{we denote} \quad
\polN_2:=\left(
\begin{matrix} n_{11}& n_{12} \\ 
n_{12} & n_{22} \end{matrix}  \right).
\end{align*}
Given row vectors $\bX,\bY\in \polR^d$,
the quadratic form $(\bX,\bY)\SCAL\polN \SCAL(\bX,\bY)^T$,
generated by the $2{\times}2$ block matrix $\polN$, is negative
semi-definite if and only if $\polN_2$ is negative
semi-definite, \ie $n_{22}\le 0$ and
$\text{det}(\polN_2)\le 0$.

\begin{lemma} \label{Lem:convexity_N+P}
Assume that
\eqref{def_f}-\eqref{def_l} hold.
The quadratic form $J$ is negative semi-definite if and only if
\begin{equation}
ad\, \text{\em det}(\Sigma) - \tfrac14 (d-a)^2 \rho^{-2} s_e^{2} p_e^2\ge 0.
\label{J_negative}
\end{equation}
Moreover, let $\lambda\in \Real$ such that $d(1+\lambda)=a$, then
\begin{equation}
J +\lambda d \frac{\rho}{s_e}\GRAD s_e \SCAL \GRAD s \le  0. 
\label{J+lambda_negative}
\end{equation}
The inequality \eqref{J+lambda_negative} becomes strict if $a>0$ and
$d>0$.
\end{lemma}%
\begin{proof}
Using the definition of $\bl$, we re-write $\bJ$ in the following form:
\begin{multline*}
\bJ = -a s_e \GRAD\rho \GRAD e - a e \GRAD\rho\GRAD s_e  
+a s_\rho |\GRAD\rho|^2 + a \rho \GRAD\rho\GRAD s_\rho
+ a e \GRAD\rho \GRAD s_e  - a \rho s_\rho \GRAD \rho \GRAD s_e\\
+ d\rho s_e^{-1}\GRAD s_e (s_\rho\GRAD \rho +s_e\GRAD e)
+ a \GRAD\rho (s_\rho\GRAD \rho +s_e\GRAD e). 
\end{multline*}
This expression can be further simplified as follows:
\begin{align*}
\bJ& = 2 a s_\rho |\GRAD\rho|^2 + a \rho
\GRAD\rho (s_{\rho\rho}\GRAD\rho + s_{\rho e}\GRAD e)\\ 
& \qquad + (d-a)\rho s_\rho s_e^{-1}
\GRAD \rho (s_{\rho e}\GRAD \rho + s_{ee}\GRAD e)
 + d \rho \GRAD e (s_{\rho e}\GRAD \rho + s_{ee}\GRAD e)\\
& = (\GRAD \rho,\GRAD e)^T \polN (\GRAD \rho,\GRAD e),
\end{align*}
where the matrix $\polN$ is defined by
\begin{align*}
\polN= & \left(\begin{matrix} n_{11} \polI_d& n_{12}\polI_d \\ 
n_{12}\polI_d & n_{22}\polI_d\end{matrix}\right);\qquad 
\begin{aligned}
n_{11} &= (d-a) \rho s_\rho s_e^{-1} s_{\rho e} + a \rho^{-1}
\partial_\rho(\rho^2 s_\rho),\\ 
2 n_{12} & = (d-a) \rho s_\rho
s_e^{-1} s_{e e} + (d +a)\rho s_{\rho e},\\ 
n_{22} & = d \rho s_{ee}.
\end{aligned}
\end{align*}

Let us define the $2{\times}2$ block matrix $\polQ$ obtained by setting $a=0$ and $d=1$
in $\polN$:
\begin{align*}
q_{11} &
= \rho s_\rho s_e^{-1} s_{\rho e},\qquad
q_{12}  = \rho s_\rho s_e^{-1} s_{e e} +  \rho s_{\rho e},\qquad
q_{22}  = \rho s_{ee}.
\end{align*}
Notice that this definition implies that the quadratic form induced by $\polQ$ 
is
\[
(\GRAD \rho,\GRAD e) \SCAL \polQ\SCAL (\GRAD \rho,\GRAD e)^T =
\frac{\rho}{s_e}\GRAD s_e\SCAL \GRAD s.
\]
Now let us consider the following $2{\times}2$ block matrix $\polM =
\polN + \lambda d \polQ$ where $\lambda\in \Real$. Let us set
$d'=d(1+\lambda)$ and observe that
\begin{align*}
m_{11} &
= (d'-a) \rho s_\rho s_e^{-1} s_{\rho e} + a \rho^{-1} \partial_\rho(\rho^2 s_\rho),\\ 
2 m_{12} &
= (d'-a) \rho s_\rho s_e^{-1} s_{e e} + (d' +a)\rho s_{\rho e},\\ 
m_{22} & = d' \rho s_{ee}.
\end{align*}
Observe finally that $J +\lambda d \frac{\rho}{s_e}\GRAD s_e \SCAL
\GRAD s =(\GRAD \rho,\GRAD e) \SCAL \polM\SCAL (\GRAD \rho,\GRAD
e)^T$.

To have a negative semi-definite form we need $m_{22} = d' \rho
s_{ee}\le 0$, which means $0\le d'$ since $s_{ee}<0$ owing to the
convexity assumption \eqref{convexity}. We also need $\text{det}(\polM_2)$ to
be non-negative,
\begin{align*}
\text{det}(\polM_2)  = &\ 
((d'-a) \rho s_\rho s_e^{-1} s_{\rho e} + a \rho^{-1} \partial_\rho(\rho^2 s_\rho))d'\rho s_{ee}\\
& -\tfrac14 ((d'-a) \rho s_\rho s_e^{-1} s_{ee} + (d' +a)\rho s_{\rho e})^2\\
= &\ a d' \left(\partial_\rho(\rho^2 s_\rho) s_{ee}- \rho^2 s_{\rho e}^2\right) 
- \tfrac14 (d'-a)^2 \rho^2 s_e^{-2}( s_{e}  s_{\rho e}-s_\rho s_{ee})^2.
\end{align*}
Now if we set $\lambda$ so that $d'=d(1+\lambda)=a$, then
$\text{det}(\polM_2)$ is non-negative and $d'=a\ge 0$.  Note in
passing that upon setting $\lambda=0$, this computation shows that
$J\le 0$ if and only if \eqref{J_negative} holds.
\end{proof}

\begin{remark}
Note that we could avoid invoking the convexity of the entropy in the
above argument by taking $a=0$ and $\lambda=-1$. This would however
defeat the purpose of our enterprise whose primary goal is to find a
nonzero viscous regularization of the mass conservation equation that
ensures positivity of the density and is entropy compatible.
\end{remark}

\begin{remark}
  Note that $J<0$ when $a=d$.
\end{remark}

\begin{theorem}[Minimum Entropy
  Principle] \label{Thm:minimum_entropy_principle} Assume that
  $\rho_0$ and $e_0$ are constant outside some compact set. Assume
  also that \eqref{def_ghl}-\eqref{def_f}-\eqref{def_l} hold.  Assume
  that the solution to
  \eqref{eq:euler_regul}--\eqref{eq:euler_regul_energy} is smooth,
  then the minimum entropy principle holds,
\[
\inf_{\bx\in\Real^d} s(\bx,t) \ge \inf_{\bx\in\Real^d} s_0(\bx),\qquad
\forall t\ge 0.
\]
\end{theorem}%
\begin{proof}
We re-write \eqref{entropy_step_one} as follows:
\begin{align*}
  \rho (\partial_t s + \bu \SCAL \GRAD s) & + \DIV((e s_e -\rho
  s_\rho)\bef - s_e \bl)
  -\bef\SCAL\GRAD(e s_e -\rho s_\rho) 
     + \bl\SCAL\GRAD s_e - s_e \polG{:}\GRAD\bu=0.
\end{align*}
Upon using \eqref{def_l} we obtain
\begin{align*}
  \rho (\partial_t s + \bu \SCAL \GRAD s)  - \DIV(d \rho\GRAD s) 
   -\bef\SCAL\GRAD(e s_e -\rho s_\rho) 
     + \bl\SCAL\GRAD s_e - s_e \polG{:}\GRAD\bu=0.
\end{align*}
Let $N:=-\bef\SCAL\GRAD(e s_e -\rho s_\rho)+ \bl\SCAL\GRAD s_e$.
Using definition \eqref{def_of_J}  we have $N = J - a\GRAD\rho\SCAL\GRAD
s$ and
\begin{equation}
  \rho (\partial_t s + \bu \SCAL \GRAD s) - \DIV(d\rho\GRAD s)  
- a \GRAD\rho\SCAL\GRAD s=
  - J  + s_e \polG{:}\GRAD\bu \ge 0.
 \label{entrop_transport}
\end{equation}
Owing to Lemma~\ref{Lem:convexity_N+P} there is $\lambda\in\Real$ so that
$J + \lambda d \frac{\rho}{s_e}\GRAD s_e\SCAL\GRAD s \le 0$.
Finally we have proved that
\begin{equation}
\begin{aligned}
  \rho (\partial_t s + \bu \SCAL \GRAD s) - \DIV(d\rho\GRAD s)  
& - (a \GRAD\rho  + \lambda d \frac{\rho}{s_e}\GRAD s_e)\SCAL\GRAD s \\ 
& = - J - \lambda d \frac{\rho}{s_e}\GRAD s_e\SCAL\GRAD s  + s_e \polG{:}\GRAD\bu \ge 0.
\end{aligned}\label{entrop_transport_fake}
\end{equation}

By assumption all the fields are smooth and $s$ is constant outside
some compact set (\ie $\rho$ and $e$ are constant outside some
time-dependent compact set since the initial data are constant outside
a compact set and the speed of propagation is finite).  For each
time $t$, $s$ reaches its minimum; let $x_{\min}(t)$ be one point where
the minimum of $s$ is reached, then $\GRAD s(x_{\min}(t),t)=0$ and $\LAP
s(x_{\min}(t),t)\ge 0$. The equation \eqref{entrop_transport_fake} implies
that
\[
\rho \partial_t s((x_{\min}(t),t)) - d\rho\LAP s(x_{\min}(t),t) \ge 0,
\]
which in turn implies that $\rho \partial_t s((x_{\min}(t),t))\ge 0$,
and we conclude that the minimum entropy principle holds.
\end{proof}

\begin{remark}
Note that the condition \eqref{J_negative} is not required to hold for
the minimum principle to hold.
\end{remark}

\section{Entropy inequalities}\label{Sec:Entropy}
We investigate in this section whether the regularization of the Euler
equations \eqref{eq:euler_regul}--\eqref{eq:euler_regul_energy} is 
compatible with some or all generalized entropy inequalities.

\subsection{Generalized entropies}
Let us consider all the generalized entropy identified in
\cite{MR1655839}. A function $\rho f(s)$ is called a
generalized entropy if $f$ is twice differentiable and
\begin{equation}
  f'(s) > 0,\qquad f'(s) c_p^{-1} - f''(s) > 0,
\qquad \forall (\rho,e) \in \Real_+^2, \label{def_of_f}
\end{equation}
where $c_p(\rho,e)=T\partial_T s(p,T)$ is the specific heat at
constant pressure. It is shown in \cite{MR1655839} that $-\rho f(s)$
is strictly convex if and only if \eqref{def_of_f} holds, \ie \eqref{def_of_f}
characterizes the maximal set of admissible entropies for the
compressible Euler equations that are of the form $\rho f(s)$.
 
\begin{theorem}[Entropy Inequalities]
\label{Thm:entropy_inequality}
Assume that \eqref{def_ghl}-\eqref{def_f}-\eqref{def_l} hold.  Any
weak solution to the regularized system
\eqref{eq:euler_regul}-\eqref{eq:euler_regul_energy} satisfies the
entropy inequality
\begin{align}
   \partial_t( \rho f(s)) + \DIV\big(\bu\rho f(s) -d \rho \GRAD f(s) - a f(s)\GRAD\rho\big) \ge 0.
\label{generalized_entropy_inequality}
\end{align}
for all generalized entropies $\rho f(s)$ if and only if $a=d$.
\end{theorem}%
\begin{proof}
Let us multiply \eqref{entrop_transport} by $f'(s)$,
\begin{align*}
  \rho (\partial_t f(s) + \bu \SCAL \GRAD f(s)) -
  \DIV(d\rho\GRAD f(s)) & + d\rho f''(s)|\GRAD s|^2
  -a f'(s) \GRAD\rho\SCAL\GRAD s \\ & +  J f'(s) = f'(s)
  s_e \polG{:}\GRAD\bu.
\end{align*}
We now multiply the mass conservation equation \eqref{eq:euler_regul} by
$f(s)$ and we add the result to the above equation:
 \begin{align*}
 \partial_t( \rho f(s)) + \DIV(\bu\rho f(s))
   & -\DIV(d \rho \GRAD f(s) + a f(s)\GRAD\rho)  \\ & +
   d\rho f''(s)|\GRAD(s)|^2 
 + J f'(s) = f'(s) s_e \polG{:}\GRAD\bu
\end{align*}
We now investigate the sign of the quantity $d\rho f''(s)|\GRAD
s|^2+Jf'(s)$.  

Owing to \eqref{def_of_f}, we have 
\begin{equation}
d\rho f''(s)|\GRAD
s|^2+Jf'(s) < (d\rho c_p^{-1} |\GRAD s|^2+J)f'(s).
\end{equation} 
We now need to determine the sign of the quadratic form in the right
hand side of the above inequality:
\begin{align*}
d\rho c_p^{-1} &|\GRAD s|^2+J
 = d\rho
c_P^{-1}|s_\rho \GRAD \rho + s_e \GRAD e|^2 + J \\
&= d\rho c_P^{-1}(s_\rho^2 |\GRAD\rho|^2 + 2 s_\rho s_e \GRAD\rho\SCAL \GRAD e + s_e^2 |\GRAD e|^2) +
J
= d \rho (\GRAD \rho,\GRAD e) \SCAL \polS \SCAL (\GRAD \rho,\GRAD e)^T,
\end{align*}
where the coefficients of the $2{\times}2$ block matrix $\polS$ are defined as follows:
\begin{align*}
d s_{11} & =d c_P^{-1}s_\rho^2+ \left((d-a) s_\rho
s_e^{-1} s_{\rho e} + a \rho^{-2} \partial_\rho(\rho^2
s_\rho)\right)\\ 
2d s_{12} & = 2dc_P^{-1} s_\rho s_e
+\left((d-a)s_\rho s_e^{-1} s_{e e} + (d +a)s_{\rho
  e}\right)\\ 
d s_{22} & = d (c_P^{-1}s_e^2 +  s_{ee}),
\end{align*}
and can be re-written into the following form
\begin{align*}
d s_{11} & =d(c_P^{-1}s_\rho^2 + \rho^{-2}\partial_\rho(\rho^2s_\rho)) 
+ (d-a) s_e^{-1} \left(s_\rho
s_{\rho e} - s_e \rho^{-2} \partial_\rho(\rho^2
s_\rho)\right)\\ 
2d s_{12} & = 2d (c_P^{-1} s_\rho s_e + s_{\rho e} )
+(d-a) s_e^{-1} \left(s_\rho s_{e e} -  s_es_{\rho
  e}\right)\\ 
d s_{22} & = d (c_P^{-1}s_e^2 +  s_{ee}).
\end{align*}
Then upon setting $x=1-\frac{a}{d}$ we infer that
\begin{align}
s_{11}  = h_{11}
+ x\rho^{-2} s_e p_\rho,\qquad
2s_{12}  = 2 h_{12}
+ x\rho^{-2} s_e p_e, \qquad
s_{22}  = h_{22} \label{def_of_S}
\end{align}
where the $2{\times}2$ matrix $\polH_2$ is defined by 
\[
\polH_2=\left(\begin{matrix} s_\rho^2 c_P^{-1} + \rho^{-2}\partial_\rho(\rho^2s_\rho)  &
s_\rho s_e c_P^{-1} +  s_{\rho e}  \\ 
s_\rho s_e c_P^{-1} + s_{\rho e}  &
s_{e}^2 c_P^{-1} + s_{ee} 
\end{matrix}\right)
\]
is shown to be negative in Lemma~\ref{Lem:H}. In particular we have
$s_{22} = h_{22}= s_{e}^2 c_P^{-1} + s_{ee} <0$ owing to the
inequality $c_p T_e > 1$ established in \eqref{c_p_T_e}. As a result,
the matrix $\polS$ is negative semi-definite if and only if
the determinant of $\polS_2$ is non-negative,
\begin{align*}
\text{det}(\polS_2)= h_{11}h_{22} & + x h_{22}\rho^{-2} s_e p_\rho -
(h_{12} + \tfrac12 x\rho^{-2} s_e p_e)^2 \\ &= \text{det}(\polH_2) +
x \rho^{-2} s_e (h_{22} p_\rho - h_{12} p_e) - \tfrac14 x^2
\rho^{-4} s_e^2 p_e^2.
\end{align*}
According to Lemma~\ref{Lem:H} we have $\text{det}(\polH_2)$ and
$h_{22} p_\rho - h_{12} p_e=0$. This proves that
\begin{align*}
\text{det}(\polS_2) &= - \tfrac14 x^2
\rho^{-4} s_e^2 p_e^2.
\end{align*}
In conclusion, $\polS$ is negative semi-definite if and only if $x=0$,
ie $a=d$.

The above argument shows that $d\rho f''(s)|\GRAD s|^2 + J f'(s)< 0$
if $a=d$ . This proves that all the generalized entropy inequalities are
satisfied if $a=d$.

If $a\not=d$ we consider generalized entropies such that $f''(s) =
(1-\epsilon)f'(s) c_p(s,\rho)$, $\epsilon\in (0,1)$ (it is always
possible to solve this ODE for any fixed value of $\rho$). 
For this
subclass of generalized entropies, 
we have
\begin{equation}
d\rho f''(s)|\GRAD
s|^2+Jf'(s) =((1-\epsilon)d\rho c_p^{-1} |\GRAD s|^2+J)f'(s). \label{bad_entropies}
\end{equation}
From the proof of Theorem\eqref{Thm:entropy_inequality}, we know that
the quadratic form $d\rho c_p^{-1} |\GRAD s|^2+J =d\rho
(\GRAD\rho,\GRAD e)\SCAL\polS(\rho,e)\SCAL(\GRAD\rho,\GRAD e)^T$ is
negative semi-definite if and and only $a=d$.  Let $(\rho^*,e^*)\in
\polR^2_+$ be a pair of positive numbers so that $a(\rho^*,e^*)\not=
d(\rho^*,e^*)$. Since the quadratic form generated by
$\polS(\rho^*,e^*)$ is not negative semi-definite, there exists a pair
of row vectors $\bX,\bY\in \polR^d$ so that
$(\bX,\bY)\SCAL\polS(\rho^*,e^*)\SCAL(\bX,\bY)^T>0$. It is always possible to
choose $\epsilon$ small enough so that
\[
(\bX,\bY)\SCAL\polS(\rho^*,e^*)\SCAL(\bX,\bY) -\epsilon d^* \rho^*
(c_p^*)^{-1} |s_{\rho}^* \bX + s_e^* \bY|^2 f'(s^*) >0.
\]
Now we define an initial state so that in the neighborhood of the
origin we have the following data: $\bm_0=0$, $\rho_0(\bx) = \rho^* +
\bx\SCAL \bX$, $e_0(\bx) = e^* + \bx\SCAL \bY$.  Notice that with this
choice $\GRAD\bu_0 =0$, $\GRAD\rho_0 = \bX$ and $\GRAD e_0 = \bY$;
therefore $d \rho_0 f''(s_0)|\GRAD(s_0)|^2 + J(\rho_0,e_0) f'(s_0) -
f'(s_0) s_e(\rho_0,e_0)\polG{:}\GRAD\bu_0 > 0$, which proves that the
entropy inequality is violated at the origin close to the initial
time. In conclusion $a=d$ is a necessary condition for all the
generalized entropy inequalities to be satisfied.
\end{proof}

\begin{remark}
  Upon re-defining the velocity $\widetilde\bu = \bu + (d-a)\GRAD
  \log\rho$, the entropy inequality
  \eqref{generalized_entropy_inequality} can be re-written into the
  following form
\begin{align}
   \partial_t( \rho f(s)) + \DIV (\widetilde\bu\rho f(s))
 -\DIV( d \rho \GRAD \rho f(s)) \ge 0.
\end{align}
\end{remark}

\begin{remark}
Theorem \ref{Thm:entropy_inequality} proves that the family of
regularization such that $a=d$ is the most robust in the sense that it
is the most dissipative. This result suggests that the choice $a=d$
may be a very good candidate to construct a robust first-order
numerical method for solving the compressible Euler equations.
\end{remark}

\begin{corollary} \label{Cor:generalized_entropy} Let $\alpha$ be a
  real number, $\alpha<1$, and assume that
  \eqref{def_ghl}-\eqref{def_f}-\eqref{def_l} hold.  Any weak
  solution to the regularized system
  \eqref{eq:euler_regul}-\eqref{eq:euler_regul_energy} satisfies the
  entropy inequality \eqref{generalized_entropy_inequality} for all
  the generalized entropies $\rho f(s)$ such that $f'>0$ and $\alpha
  c_p^{-1} f' \ge f''$ if $2\Gamma - 2\Delta^{\frac12} <
  1-\frac{a}{d} < 2\Gamma + 2\Delta^{\frac12}$ where $\Gamma =
  (1-\alpha) \text{det}(\Sigma) \rho^2 s_e^{-2} p_e^{-2}$ and $\Delta
  = \Gamma(1+\Gamma)$.
\end{corollary}%
\begin{proof} We proceed as in the proof of
  Theorem~\ref{Thm:entropy_inequality} where we replace $\polH$ by
  $\polH^\alpha$ where $c_p^{-1}$ is substituted by $\alpha
  c_p^{-1}$. Upon replacing $c_p^{-1}$ by $\alpha c_p^{-1}$ in the
  proof of Lemma~\ref{Lem:H}, we infer that $\text{det}(\polH_2^\alpha)
  = (1-\alpha) \rho^{-2} \text{det}(\Sigma) $ and
  $s_e(h^\alpha_{22} p_\rho - h^\alpha_{21} p_e ) = (1-\alpha)
  \rho^{-2} \text{det}(\Sigma)$. Then by defining $\polS^\alpha$ as in
  \eqref{def_of_S}, where $\polH$ is substituted by $\polH^\alpha$, we
  obtain
\begin{align*}
  \text{det}(\polS_2^\alpha) &= (1-\alpha) \rho^{-2} \text{det}(\Sigma) + x
  \rho^{-2} (1-\alpha) \text{det}(\Sigma) - \tfrac14 x^2
  \rho^{-4} s_e^2 p_e^2\\
  &= \rho^{-2}((1-\alpha) \text{det}(\Sigma) (1+x)- \tfrac14 x^2
  \rho^{-2} s_e^2 p_e^2),
\end{align*}
where we defined $x=1-\frac{a}{d}$.  Then upon setting $\Gamma =
(1-\alpha) \text{det}(\Sigma) \rho^2 s_e^{-2} p_e^{-2}$ and $\Delta =
\Gamma(1+\Gamma)$, we conclude that the matrix $\polS^\alpha$ is
negative definite if
\[
2\Gamma - 2\Delta^{\frac12} < 1-\frac{a}{d} < 2\Gamma + 2\Delta^{\frac12},
\]
which ends the proof.
\end{proof}


\begin{corollary} \label{Cor:physical_entropy}
  Any weak solution to the regularized system
  \eqref{eq:euler_regul}-\eqref{eq:euler_regul_energy} satisfies the
  entropy inequality \eqref{generalized_entropy_inequality} for the
  physical entropy $\rho s$ (\ie $f(s)=s$) if $2\Gamma -
  2\Delta^{\frac12} < 1-\frac{a}{d} < 2\Gamma + 2\Delta^{\frac12}$
  where $\Gamma =\text{det}(\Sigma) \rho^2 s_e^{-2}
  p_e^{-2}$ and $\Delta = \Gamma(1+\Gamma)$.
\end{corollary}

\begin{proof}
Take $\alpha=0$ in Corollary~\ref{Cor:generalized_entropy} or use \eqref{J_negative}.
\end{proof}

\subsection{Ideal gas}
Let us illustrate the above theory in the case of ideal gases, \ie
$s=\log(e^{\frac{1}{\gamma-1}}\rho^{-1})$ with $\gamma>1$.  We have
$c^2=\gamma(\gamma-1)e$, $c_p=\gamma(\gamma-1)^{-1}$,
$\text{det}(\Sigma)=(\gamma-1)^{-1} e^{-2}$, $\bef = a\GRAD \rho$, and
$\bl = \gamma d e (\frac{a}{d}-1+\frac{1}{\gamma})\GRAD\rho + d
\rho\GRAD e$.
The range for the ratio
$ad^{-1}$ for Corollary~\ref{Cor:physical_entropy} to hold is
\begin{equation}
  \frac{2}{\gamma-1}(1-\sqrt{\gamma}) < 1- \frac{a}{d} < \frac{2}{\gamma-1}(1+\sqrt{\gamma}).
\end{equation}
In particular the choice $1- \frac{a}{d} = \frac{1}{\gamma}$ is
clearly in the admissible range for the physical entropy inequality.
This particular choice is such that $\bl = d\rho \GRAD e$
and $\bef = d \frac{\gamma-1}{\gamma}\GRAD\rho$ , \ie $\bl$
does involve any mass dissipation.

\section{Discussion}
\label{Sec:Brenner}
We show in this section that the regularization proposed above is a
bridge between the Navier-Stokes and parabolic regularizations of the
Euler equations that reconciles the two point of views.
\subsection{Parabolic regularization}\label{Sec:parabolic_regul}
The first natural question that comes to mind is how different is the
general regularization
\eqref{eq:euler_regul}-\eqref{eq:euler_regul_energy} from other known
regularizations. In particular how does it differ from the parabolic
regularization \eqref{eq:Para}-\eqref{eq:Para_init}?  The answer is
given by the following, somewhat a priori frustrating result:
\begin{proposition}
  The parabolic regularization \eqref{eq:Para}-\eqref{eq:Para_energy}
  is identical to \eqref{eq:euler_regul}-\eqref{eq:euler_regul_energy}
  with \eqref{def_ghl}-\eqref{def_l} where $a=d=\epsilon$, $\polG =
  \epsilon\rho\GRAD\bu$.
\end{proposition}

\begin{proof}
The equality $a=\epsilon$ comes from the identification
$\bef=\epsilon\GRAD\rho$ in the mass conservation equation in
\eqref{eq:Para} and \eqref{eq:euler_regul}.  The identity
$\epsilon\GRAD\bm = \epsilon\GRAD\rho\otimes\bu + \epsilon\rho\GRAD
\bu$ implies that upon setting $\polg = \bef\otimes \bu + \polG$ with $\polG
= \epsilon\rho\GRAD\bu$, the momentum conservation equations in
\eqref{eq:Para_moment} and \eqref{eq:euler_regul_moment} are
identical.  Upon observing that
\begin{align*}
\polg\SCAL \bu = \bu^2 \bef + \polG\SCAL\bu=\epsilon \bu^2\GRAD\rho
+ \frac12 \epsilon \rho\GRAD\bu^2 = \epsilon\GRAD\frac12\rho\bu^2 +
\frac12 \bu^2 \bef,
\end{align*}
we obtain that
\[
\epsilon\GRAD\bE = \epsilon \GRAD (\rho e) +
\GRAD\frac12 \epsilon \rho \bu^2 =  \epsilon \GRAD (\rho e)  -\frac12 \bu^2 \bef + \polg\SCAL \bu.
\]
The energy equations in \eqref{eq:Para_energy} and
\eqref{eq:euler_regul_energy} are identical if one sets
$\bh=\bl-\frac12 \bu^2 \bef$, with 
$\bl=\epsilon \GRAD (\rho e)$, meaning $d=\epsilon$.
\end{proof}

\begin{remark}
Even when $a=d$, one important interest of the class of
regularization \eqref{eq:euler_regul}-\eqref{eq:euler_regul_energy},
when compared to the monolithic parabolic regularization, is that it
decouples the regularization on the velocity from that on the density
and internal energy.  In particular the regularization on the velocity
can be made rotation invariant by making the tensor $\polG$ a function
of the symmetric gradient $\GRAD^s\bu$. This decoupling was not a
priori evident (at least to us) when looking at the monolithic
parabolic regularization \eqref{eq:Para}-\eqref{eq:Para_energy}.
\end{remark}

\subsection{Connection with phenomenological models}
When introducing the structural assumptions
\eqref{def_ghl}-\eqref{def_l} in the balance equations
\eqref{eq:euler_regul}-\eqref{eq:euler_regul_energy} we obtain the
following system:
\begin{align}\label{eq:euler_regul_ghl_mass}
  &\partial_t \rho + \DIV \bm - \DIV \bef = 0,
  \\ \label{eq:euler_regul_ghl_moment} &\partial_t \bm + \DIV
  (\bu\otimes \bm) + \GRAD p- \DIV (\polG(\GRAD^s\bu) + \bef\otimes \bu) =
  0, \\ &\partial_t E + \DIV (\bu(E+p)) - \DIV (\bl + \tfrac12 \bu^2
  \bef + \polG(\GRAD^s\bu)\SCAL \bu) = 0,\label{eq:euler_regul_ghl_energy}
\end{align}
 When looking at
 \eqref{eq:euler_regul_ghl_mass}-\eqref{eq:euler_regul_ghl_energy} it
 is not immediately clear how this system can be reconciled either
 with the Navier-Stokes regularization or with any phenomenological
 modeling of dissipation.

It is remarkable that this exercise can actually been done by
introducing the quantity $\bu_m=\bu - \rho^{-1}\bef$. The conservation
equations then becomes
\begin{align}
&\partial_t \rho + \DIV(\bu_m\rho) = 0, \label{Brenner_mass}\\
&\partial_t \bm + \DIV (\bu_m\otimes \bm)+\GRAD p-\DIV (\polG(\GRAD^s\bu)) = 0,\label{Brenner_momentum}\\
&\partial_t E + \DIV (\bu_mE) - \DIV (\bl - e \bef)
 +\DIV \big((p\polI -\polG(\GRAD^s\bu))\SCAL \bu\big)=0, \label{Brenner_energy}
\end{align}
with again $\bm=\rho\bu$ and $E=\rho e + \frac12 \rho\bu^2$. It is
surprising that this system involves two velocities.  It is also
somewhat surprising to observe that the above system resembles the
Navier-Stokes regularization. In particular if one sets $a=d$, the
term $\bl - e \bef$ becomes $d\rho\GRAD e$, which upon assuming $\diff
e=c_v \diff T$, reduces to $d(\rho,e) \rho c_v \GRAD T$, \ie one
obtains Fourier's law: $\bl - e \bef=d(\rho,e) \rho c_v \GRAD T$.  

During the preparation of this paper, it has been brought to our
attention that the regularization model that we propose above somewhat
resembles, at least formally, a model of fluid dynamics of
\cite{Brenner2006190} (see \eg equations (1) to (5) in
\cite{Brenner2006190}). The author has derived the above system of
conservation equations (up to some non-essential disagreement on the
term $\bl -e\bef$) by invoking theoretical arguments from
\cite{ChristianOttinger:2005p5495} and phenomenological
considerations.  The mathematical properties of this system have been
investigated thoroughly by \cite{MR2732009}. Brenner has been
defending the idea that it makes phenomenological sense to distinguish
the so-called mass velocity, $\bu_m$, from the so-called volume
velocity, $\bu$, since 2004 (or so). We do not want to enter this
debate, but this idea seems to be supported by our mathematical
derivation
\eqref{Brenner_mass}-\eqref{Brenner_energy} that
did not invoke any had oc phenomenological assumption. Recall that our
primal motivation in this project is to find a regularization of the
compressible Euler equations that can serve as a good numerical
device, and by being good we mean that the model must give positive
density, positive internal energy, a minimum entropy principle and be
compatible with a large class of entropy inequalities.

\subsection{Conclusions}\label{Sec:Conclusions}
Let us finally rephrase our findings. In its most general form, the
regularized system \eqref{Brenner_mass}-\eqref{Brenner_energy} can be
re-written as follows:
\begin{align}
  &\partial_t \rho + \DIV(\bu_m\rho ) = 0,\label{conclusion:mass} \\
  &\partial_t \bm + \DIV (\bu_m\otimes \bm) + \GRAD p- \DIV
  (G(\GRAD^s\bu)) = 0,\\ &\partial_t E + \DIV (\bu_mE) - \DIV \bq
  +\DIV \big((p\polI -G(\GRAD^s\bu))\SCAL \bu\big)=0\\ &\bu_m=\bu -
  a(\rho,e) \GRAD \log\rho \\ &\bq= (a-d)p\GRAD\log\rho + d \rho
  \GRAD e,\qquad a(\rho,e) \ge 0,\ d(\rho,e)\ge
  0.\label{conclusion:f_and_q}
\end{align}
It is established in Lemma~\ref{Lem:positive_density_principle} that
the definition of $\bef = a(\rho,e) \GRAD\rho$ is compatible with the
positive density principle. The particular form of $\bq$ in
\eqref{conclusion:f_and_q} results from the definition 
of $\bl$, see \eqref{def_l},
which is required for the minimum entropy principle to hold, as
established in Theorem~\ref{Thm:minimum_entropy_principle}.  It is
finally proved in Theorem~\ref{Thm:entropy_inequality} that the most
robust regularization, \ie that which is compatible with all the
generalized entropy \`a la \cite{MR1655839}, corresponds to the choice
$a=d$. Various relaxations of the constraint $a=d$ are described in
Corollary~\ref{Cor:generalized_entropy} and in
Corollary~\ref{Cor:physical_entropy}. As observed in
\S\ref{Sec:parabolic_regul}, the parabolic regularization can be put
into the form \eqref{conclusion:mass}-\eqref{conclusion:f_and_q} with
the particular choice $\polG = a\GRAD\bu$, which is not rotation
invariant and uses the same viscosity coefficient for all fields.

\begin{appendix}
\section{Primer in thermodynamics}\label{Sec:Appendix}
We collect in this appendix standard results from thermodynamics that
are used in the paper.
\subsection{Chain rule}
Let $\Phi : \Real^2\ni (\alpha,\beta)
\longmapsto\Phi(\alpha,\beta)=(\phi(\alpha,\beta),\psi(\alpha,\beta))\in
\Real^2$ be a $\calC^1$-diffeomorphism.  The following holds:
\begin{equation}
\frac{1}{\partial_\alpha\phi\partial_\beta\psi -\partial_\beta\phi\partial_\alpha\psi}
\left(\begin{matrix}
\partial_\beta \psi & -\partial_\beta \phi \\
-\partial_\alpha\psi & \partial_\alpha \phi
\end{matrix}\right)
= \left(\begin{matrix}
\partial_\phi \alpha & \partial_\psi \alpha \\
\partial_\phi\beta & \partial_\psi \beta
\end{matrix}\right). \label{two_variable}
\end{equation}
In particular if $\phi(\alpha,\beta)=\alpha$ we have
\begin{equation}
  \partial_\alpha\beta(\alpha,\psi) = 
- \frac{\partial_\alpha \psi(\alpha,\beta)}{\partial_\beta \psi(\alpha,\beta)},
\qquad
\partial_\psi\beta(\alpha,\psi) = \frac{1}{\partial_\beta \psi(\alpha,\beta)} \label{one_variable}
\end{equation}

\subsection{Speed of sound}
The square of the speed of sound is defined to be 
\begin{align}
  c^2 := \partial_\rho p(\rho,s), \label{def_1_c_square}
\end{align} 
\ie $c^2$ is the partial derivative of the pressure as a
function of the density and the specific entropy. Using the chain
rule, this definition is equivalent to
\begin{equation}
c^2 = \partial_\rho p(\rho,s) = \partial_\rho p(\rho,e) + \partial_e
p(\rho,e) \partial_\rho e(\rho,s),
\end{equation}
and using \eqref{one_variable} with $\alpha=\rho$, $\beta=e$,
$\psi=s$, one obtains
\begin{equation}
c^2 =  p_\rho  - \frac{s_\rho}{s_e} p_e(\rho,e). \label{def_2_c_square}
\end{equation}
Using the following representations of $p_e$ and $p_\rho$:
\begin{align}
p_e=\rho^2 s_e^{-2}(s_\rho
s_{ee} - s_e s_{\rho e}),  \qquad
p_\rho = s_e^{-2}(\rho^2 s_\rho s_{\rho e} - s_e \partial(\rho^2 s_\rho)),
\label{thermodynamic_compatibility}
\end{align}
the expression \eqref{def_2_c_square} also gives
\begin{equation}
c^2 = \rho^2 s_e^{-3}(2 s_e s_\rho s_{\rho e} - s_e^2 \rho^{-2} \partial(\rho^2 s_\rho)
- s_\rho^2 s_{ee}). \label{unilluminating_c_square}
\end{equation}

\subsection{Convexity of the entropy, $\text{det}(\Sigma)$}
Let us define the following matrix
\begin{equation}
\Sigma := \rho \left(\begin{matrix}\rho^{-2} \partial_\rho (\rho^2 s_{\rho}) &  s_{\rho e}\\ 
 s_{\rho e} & s_{ee}\end{matrix}\right), \label{A_def_of_Sigma}
\end{equation}
which, up to the $\rho$ factor, is the Hessian of the entropy with
respect to the variables $(\rho^{-1},e)$.  The convexity assumption on
the entropy implies that $s_{ee}$ and $\rho^{-1} \partial_\rho (\rho^2
s_{\rho})$ are negative.  We have the following characterization of
the determinant of $\Sigma$.
\begin{equation}
\text{det}(\Sigma) = s_e^3(p_\rho T_e - p_e T_\rho). \label{det_sigma}
\end{equation}
To prove this statement, we observe that the following holds owing to
\eqref{thermodynamic_compatibility}:
\begin{align*}
s_e^2 T_e &= -s_{ee}, &&  s_e^2 T_\rho = - s_{\rho e},\\
s_e^2 p_e & = \rho^2 (s_\rho s_{ee} - s_e s_{\rho e}) && 
s_e^2 p_\rho = \rho^2\left(s_\rho s_{\rho e} - s_{e}\rho^{-2}\partial_\rho(\rho^2 s_\rho)\right).
\end{align*}
The result is now evident.

\subsection{Specific heat at constant pressure }
The specific heat at constant pressure is defined to be
$c_p(\rho,e)=T\partial_T s(T,p)$.
\begin{lemma}
  The quantities $\text{\em det}(\Sigma)$, $c^2$ and $c_p$ are related by
\begin{equation}
c_p\text{\em det}(\Sigma)  = s_e^3 c^2. \label{c_p_c_square_det_sigma}
\end{equation}
\end{lemma}%
\begin{proof}
Using the chain rule, we can re-write the above definition as follows:
\begin{align*}
c_p(\rho,e)= s_e^{-1} (s_\rho \rho_T(p,T) +  s_e  e_T(p,T)).
\end{align*} 
The change of variable formula \eqref{two_variable}
with the convention $(\alpha=\rho,\beta=e)$ and
$(\phi=p,\psi=T)$ gives
\begin{equation*}
\rho_T(p,T) =
\frac{-p_e}{p_\rho T_e - p_e T_\rho},\qquad
e_T(p,T) =
\frac{p_\rho}{p_\rho T_e - p_e T_\rho}.
\end{equation*}
We then have the following expression for $c_p$
\begin{align}
c_p = s_e^{-1} \frac{( p_\rho s_e - p_e s_\rho)}{p_\rho T_e - p_e T_\rho}.
\label{c_p_bis}
\end{align} 
Then using the expression of $c^2$ in \eqref{def_2_c_square} and the
relation \eqref{det_sigma}, we arrive at the desired expression.
\end{proof}

\begin{lemma}
The following holds:
\begin{equation}
c_p T_e >1. \label{c_p_T_e}
\end{equation}
\end{lemma}%
\begin{proof}
  The definition of $c_p$ implies that we need to estimate $T s_T(p,T)
  T_e(\rho,e)$. The chain rule implies
\[
1 = T s_e(\rho,e) = T s_p(p,T) p_e(\rho,e) + T s_T(p,T) T_e(\rho,e).
\]
The result will be established if we can prove that $s_p(p,T) p_e(\rho,e)<0$.
We now calculate $s_p(p,T)$. The chain rule implies again that
\[
(s_p(p,T))^{-1} = p_s(s,T) = p_\rho(\rho,e) \rho_s(s,T) + p_e(\rho,e) e_s(s,T).
\]
Then using \eqref{two_variable}
with the convention $(\alpha=\rho,\beta=e)$ and
$(\phi=s,\psi=T)$ gives
\[
\rho_s(s,T) = \frac{T_e}{s_\rho T_e - s_e T_\rho},\qquad
e_s(s,T) = \frac{-T_\rho}{s_\rho T_e - s_e T_\rho}.
\]
This in turn implies that
\[
(s_p(p,T))^{-1} =p_s(s,T) =  \frac{p_\rho T_e - p_e T_\rho}{s_\rho T_e - s_e T_\rho}
= - \frac{s_e^{-3}\text{det}(\Sigma)}{\rho^{-2} p_e},
\]
since $s_\rho T_e - s_e T_\rho = s_e^{-2}(-s_\rho s_{ee} + s_e s_{\rho
  e})=-\rho^{-2} p_e$, where we used
\eqref{thermodynamic_compatibility}. In conclusion $s_p(p,T)
p_e(\rho,e)=-s_e^3 p_e^2 \rho^{-2}\text{det}(\Sigma)^{-1}<0$, owing to
\eqref{convexity} and \eqref{s_e_positivity}, which concludes the
proof.
\end{proof}

\begin{remark} \label{Rem:hyperbolicity}
  Note in passing that the convexity assumption \eqref{convexity}
  implies that $T_e>0$, which owing to \eqref{c_p_T_e} implies that
  $c_p>0$. This in turn implies that $c^2>0$ owing to
  \eqref{c_p_c_square_det_sigma}, \ie the Euler system
  \eqref{eq:Euler}-\eqref{eq:Euler_init} is hyperbolic under the
  convexity assumption \eqref{convexity} and the positivity assumption
  on the temperature \eqref{s_e_positivity}. Positivity of the
  pressure is not needed to establish this fact.
\end{remark}

\subsection{Matrix $\polH_2$}
Investigations on entropy inequalities involve the quadratic form
induced by the matrix $\polH_2$
\[
\polH_2=\left(\begin{matrix} s_\rho^2 c_P^{-1} + \rho^{-2}\partial_\rho(\rho^2s_\rho)  &
s_\rho s_e c_P^{-1} +  s_{\rho e}  \\ 
s_\rho s_e c_P^{-1} + s_{\rho e}  &
s_{e}^2 c_P^{-1} + s_{ee} 
\end{matrix}\right)
\]
Some key properties of this matrix are collected in the following
lemma.
\begin{lemma} \label{Lem:H}
The following hold:
\begin{enumerate}[(i)]
\item \label{H1}$\text{\em det}(\polH_2)=0$.
\item \label{H2} $\polH_2$ is negative semi-definite.
\item \label{H3} $h_{22} p_\rho - h_{12} p_e=0$.
\end{enumerate}
\end{lemma}%
\begin{proof} 
\eqref{H1} Using the expressions \eqref{unilluminating_c_square} and
\eqref{det_sigma} for the speed of sound, $c^2$, and
$\text{det}(\Sigma)$, and the relation \eqref{c_p_c_square_det_sigma},
the determinant of $\polH_2$ is re-written as follows:
\begin{align*}
\text{det}(\polH_2) &=( s_\rho^2 c_P^{-1} +
\rho^{-2}\partial_\rho(\rho^2s_\rho) )(s_{e}^2 c_P^{-1} + s_{ee})
-(s_\rho s_e c_P^{-1} + s_{\rho e})^2 \\ &=\rho^{-2}
\text{det}(\Sigma) + c_P^{-1}
(s_e^2\rho^{-2}\partial_\rho(\rho^2s_\rho) + s_\rho^2 s_{ee} -2s_\rho
s_e s_{\rho e}) \\ & =\rho^{-2} \text{det}(\Sigma) - c_P^{-1} c^2
\rho^{-2} s_e^3 =0.
\end{align*}
This is essentially the result established in
\cite[p.~2126]{MR1655839}.  

\eqref{H2} Owing to the inequality $1<c_p T_e$
established in \eqref{c_p_T_e}, we infer that $h_{22}=
s_{e}^2 c_P^{-1} + s_{ee} <0$, which together with \eqref{H1} 
proves statement \eqref{H2}.

\eqref{H3} Let us compute $s_e^{-2}(h_{22} p_\rho  -
h_{12}p_e)$,
\begin{align*}
s_e^{-2}(h_{22} p_\rho - h_{12}p_e) &=
(c_p^{-1} - T_e) p_\rho - (s_\rho s_e^{-1} c_p^{-1} - T_\rho) p_e \\
&=   p_e T_\rho - p_\rho  T_e + c_p^{-1}s_e^{-1}(s_e  p_\rho - s_\rho p_e).
\end{align*}
This proves that $s_e^{-2}(h_{22} p_\rho - h_{12}p_e)=0$
owing to \eqref{c_p_bis}.
\end{proof}
\end{appendix}


\end{document}